\documentclass{llncs}
 
\usepackage{array}
\usepackage{multirow}
\usepackage[dvips]{graphicx}
\usepackage{amssymb,amsmath,stmaryrd}
\usepackage{latexsym}
\usepackage{subfigure}
\usepackage[usenames]{color}
\usepackage{gastex}
\usepackage{comment}
\usepackage{times}
\let\emptyset\varnothing

\def\break{\penalty-1000}

\newcounter{compressEnum}
\renewcommand{\thecompressEnum}{$\roman{compressEnum}$}
\newenvironment{compressEnum}
{\setcounter{compressEnum}{0}}
{}
\newcommand{\itCompress}{\stepcounter{compressEnum}{(\thecompressEnum) }}

\newtheorem{conj}{Conjecture} 

\newtheorem{theo}{Theorem} 

\newtheorem{lemm}{Lemma}

\newtheorem{coro}{Corollary}

\newtheorem{defi}{Definition}

\makeatletter
\let\c@lemm\c@theo
\let\c@coro\c@theo
\let\c@defi\c@theo
\let\c@assu\c@theo
\makeatother

\def\abs#1{\ensuremath{\lvert #1\rvert}}
\def\bigabs#1{\ensuremath{\big\lvert #1 \big\rvert}}

\def\norm#1{\ensuremath{\lVert #1\rVert}}

\newcommand{\rat}{{\mathbb Q}}
\newcommand{\ratl}{{\mathbb L}}
\newcommand{\nat}{\mathbb N}

\newcommand{\real}{{\mathbb R}}

\newcommand{\C}{\mathcal{C}}
\newcommand{\Z}{\mathcal{Z}}

\newcommand{\Val}{\mathsf{Val}}

\newcommand{\LimAvg}{\mathsf{LimAvg}}
\newcommand{\LimInfAvg}{\mathsf{LimInfAvg}}
\newcommand{\LimSupAvg}{\mathsf{LimSupAvg}}

\newcommand{\Avg}{\mathsf{Avg}}
\newcommand{\Sum}{\mathrm{sum}}
\newcommand{\weight}{\mathsf{wt}}
\newcommand{\conv}{\mathsf{conv}}

\def\abs#1{\ensuremath{\lvert #1\rvert}}
\newcommand{\tuple}[1]{\langle #1 \rangle}
\newcommand{\set}[1]{\{#1\}}
\newcommand{\capsp}{{\; \cap \;}}

\newcommand{\constr}{\gamma}
\newcommand{\convk}{\mathsf{conv}}
\newcommand{\ok}{\checkmark}
\newcommand{\ko}{$\times$}

\title{{\bf Mean-Payoff Automaton Expressions}} 

\author{Krishnendu Chatterjee\inst{1} \and Laurent Doyen\inst{2} \and Herbert Edelsbrunner\inst{1} \and \\ Thomas A. Henzinger\inst{1} \and Philippe Rannou\inst{3}
}

\institute{IST Austria (Institute of Science and Technology Austria) \and
LSV, ENS Cachan \& CNRS, France \and
ENS Cachan Bretagne, Rennes, France}

\begin{document}
\maketitle

\begin{abstract}
Quantitative languages are an extension of boolean languages that assign to 
each word a real number. 
Mean-payoff automata are finite automata with numerical weights on transitions
that assign to each infinite path 
the long-run average of the transition weights.
When the mode of branching of the automaton is deterministic, nondeterministic,
or alternating, the corresponding class of quantitative languages is not 
\emph{robust} as it is not closed under the pointwise operations of 
max, min, sum, and numerical complement.
Nondeterministic and alternating mean-payoff automata are not \emph{decidable}
either, as the quantitative generalization of 
the problems of universality and language inclusion is undecidable.

We introduce a new class of quantitative languages, defined by 
\emph{mean-payoff automaton expressions}, which is robust and decidable: 
it is closed under the four pointwise operations, and we show that all 
decision problems are decidable for this class.
Mean-payoff automaton expressions subsume deterministic mean-payoff automata,
and we show that they have expressive power incomparable to 
nondeterministic and alternating mean-payoff automata.
We also present for the first time an algorithm to compute distance between two 
quantitative languages, and in our case the quantitative languages are given
as mean-payoff automaton expressions.
\end{abstract}

\section{Introduction}\label{sec:intro}

Quantitative languages $L$ are a natural generalization of boolean languages
that assign to every word $w$ a real number $L(w) \in \real$ instead of a 
boolean value. 
For instance, the value of a word (or behavior) can be interpreted as the 
amount of some resource (e.g., memory consumption, or power consumption) needed
to produce it, or bound the long-run average available use of the resource. 
Thus quantitative languages can specify properties related to 
resource-constrained programs, and an implementation $L_A$ satisfies 
(or refines) a specification $L_B$ if $L_{A}(w)\leq L_{B}(w)$ for 
all words~$w$. 
This notion of refinement is a \emph{quantitative generalization of language 
inclusion}, and it can be used to check for example if for each behavior, 
the long-run average response time 
of the system lies below the specified average response requirement.
Hence it is crucial to identify some relevant class of 
quantitative languages for which this question is decidable.
The other classical decision questions such as emptiness, universality, and 
language equivalence have also a natural quantitative extension.
For example, the \emph{quantitative emptiness problem} asks,
given a quantitative language~$L$ and a threshold $\nu \in \rat$, whether there 
exists some word $w$ such that $L(w) \geq \nu$, and the 
\emph{quantitative universality problem} asks whether $L(w) \geq \nu$ for all 
words $w$. 
Note that universality is a special case of language inclusion 
(where $L_A(w) = \nu$ is constant).

Weighted \emph{mean-payoff automata} present a nice framework to express
such quantitative properties~\cite{CDH08}.
A weighted mean-payoff automaton is a finite automaton with numerical 
weights on transitions. The value of a word $w$ is the maximal value of all 
runs over~$w$ (if the automaton is nondeterministic, then there may be many 
runs over~$w$), and the value of a run $r$ is the long-run average of the
weights that appear along~$r$. A mean-payoff extension to alternating automata 
has been studied in~\cite{CDH-FCT09}.
Deterministic, nondeterministic and alternating mean-payoff automata are three 
classes of mean-payoff automata with increasing expressive power.
However, none of these classes is closed under the four pointwise operations
of max, min (which generalize union and intersection respectively), numerical 
complement\footnote{The numerical complement of a quantitative languages 
$L$ is $-L$.}, 
and sum (see Table~\ref{tab:properties}). 
Deterministic mean-payoff automata are not closed under max, min, and 
sum~\cite{CDH09b}; nondeterministic mean-payoff automata are not closed under 
min, sum and complement~\cite{CDH09b}; and alternating mean-payoff automata are 
not closed under sum~\cite{CDH-FCT09}.
Hence none of the above classes is \emph{robust} with respect to 
closure properties.

Moreover, while deterministic mean-payoff automata enjoy decidability of all
quantitative decision problems~\cite{CDH08}, 
the quantitative language-inclusion problem is undecidable
for nondeterministic and alternating mean-payoff automata~\cite{DDGRT10}, 
and thus also all decision problems are undecidable for alternating mean-payoff automata.
Hence although mean-payoff automata provide a nice framework to express 
quantitative properties, there is no known class which is both robust and 
decidable (see Table~\ref{tab:properties}).

In this paper, we introduce a new class of quantitative languages that are 
defined by \emph{mean-payoff automaton expressions}. 
An expression
is either a deterministic mean-payoff automaton, or it is the max, min, or sum
of two mean-payoff automaton expressions. 
Since deterministic mean-payoff automata are closed under complement, 
mean-payoff automaton expressions form a robust class that is closed under 
max, min, sum and complement.
We show that 
(a) all decision problems (quantitative emptiness, universality, inclusion, and
equivalence) are decidable for mean-payoff automaton expressions; 
(b) mean-payoff automaton expressions are incomparable in 
expressive power with both the nondeterministic and alternating mean-payoff 
automata (i.e., there are quantitative languages expressible by mean-payoff 
automaton expressions that are not expressible by alternating mean-payoff 
automata, and there are quantitative languages expressible by nondeterministic 
mean-payoff automata that are not expressible by mean-payoff automata 
expressions); and (c) the properties of cut-point languages (i.e., the sets of 
words with value above a certain threshold) for deterministic automata 
carry over to mean-payoff automaton expressions, mainly the cut-point language
is $\omega$-regular when the threshold is isolated (i.e., some neeighborhood
around the threshold contains no word).
Moreover, mean-payoff automaton expressions can express all examples in the
literature of quantitative properties using mean-payoff 
measure~\cite{AlurDMW09,CDH09b,CGHIKPS08}.
Along with the quantitative generalization of the classical decision problems,
we also consider the notion of \emph{distance} between two quantitative languages
$L_A$ and $L_B$, defined as $\sup_w \abs{L_A(w) - L_B(w)}$. 
When quantitative language inclusion does not hold between an implementation $L_A$
and a specification $L_B$, the distance is a relevant information 
to evaluate how close they are, as we may accept implementations
that overspend the resource but we would prefer the least expensive ones.
We present the first algorithm to compute the distance between two 
quantitative languages: we show that the distance can be computed for 
mean-payoff automaton expressions. 

\begin{table}[!tb]
\begin{center}
\begin{tabular}{|l|*{3}{c|}{l|}*{4}{c|}}
\hline
       & \multicolumn{4}{c}{Closure properties} & \multicolumn{4}{|c|}{Decision problems} \\
\cline{2-9}
       & \,$\max$\, & \,$\min$\, & \,$\Sum$\, & \,comp.\,     & \,empt.\, & \,univ.\, & \,incl.\, & \,equiv.\, \\
\hline
\,Deterministic     &  \ko & \ko  & \ko   & \quad \ok\footnotemark\addtocounter{footnote}{-1}    & \ok & \ok & \ok & \ok \\
\hline
\,Nondeterministic\,  &  \ok & \ko  & \ko   & \quad \ko      & \ok & \ko & \ko & \ko \\
\hline
\,Alternating       &  \ok & \ok  & \ko   & \quad \ok\footnotemark       & \ko & \ko & \ko & \ko \\
\hline
\,Expressions       &  \ok & \ok  & \ok   & \quad \ok      & \ok & \ok & \ok & \ok \\
\hline
\end{tabular}
\end{center}
\caption{Closure properties and decidability of the various classes of mean-payoff automata. 
Mean-payoff automaton expressions enjoy fully positive closure and decidability properties.
\label{tab:properties}}
\end{table}

\footnotetext{Closure under complementation holds
because $\LimInfAvg$-automata and $\LimSupAvg$-automata are dual. It would not hold  
if only $\LimInfAvg$-automata (or only $\LimSupAvg$-automata) were allowed.}

Our approach to show decidability of mean-payoff automaton expressions
relies on the characterization and algorithmic computation of the 
values set $\set{L_E(w) \mid w \in \Sigma^\omega}$ of 
an expression $E$, i.e. the set of all values of words according to $E$.
The value set can be viewed as an abstract representation of the 
quantitative language $L_E$, and we show that all decision problems, cut-point 
language and distance computation can be solved efficiently once we have this 
set.

First, we present a precise characterization of the value set for 
quantitative languages defined by mean-payoff automaton expressions. 
In particular, we show that it is not sufficient to construct the convex 
hull $\conv(S_E)$ of the set of 
the values of simple cycles in the mean-payoff automata occurring in $E$,
but we need essentially to apply an operator $F_{\min}(\cdot)$ which given a 
set $Z \subseteq \real^n$ computes the set of points $y \in \real^n$ that can 
be obtained by taking pointwise minimum of each coordinate of points of a set 
$X \subseteq Z$.
We show that while we need to compute the set $V_E = F_{\min}(\conv(S_E))$ to 
obtain the value set, 
and while this set is always convex, it is not always the case that 
$F_{\min}(\conv(S_E)) = \conv(F_{\min}(S_E))$ (which would immediately give
an algorithm to compute $V_E$). This may appear counter-intuitive because
the equality holds in $\real^2$ but we show that the equality does not hold 
in $\real^3$ (Example~\ref{examp2}). 

Second, we provide algorithmic solutions to compute 
$F_{\min}(\conv(S))$, for a finite set $S$.
We first present a constructive procedure that given $S$ constructs a finite 
set of points $S'$ such that $\conv(S')=F_{\min}(\conv(S))$. 
The explicit construction presents interesting properties about the set 
$F_{\min}(\conv(S))$, however the procedure itself is computationally 
expensive.
We then present an elegant and geometric construction of $F_{\min}(\conv(S))$
as a set of linear constraints.
The computation of $F_{\min}(\conv(S))$ is a new problem in computational 
geometry and the solutions we present could be of independent interest.
Using the algorithm to compute $F_{\min}(\conv(S))$, we show that all 
decision problems for mean-payoff automaton expressions are decidable.
Due to lack of space, most proofs are given in the appendix.
 
\smallskip\noindent{\em  Related works.} 
Quantitative languages have been first studied over finite words in the 
context of probabilistic automata~\cite{Rabin63}    
and weighted automata~\cite{Wautomata}.    
Several works have generalized
the theory of weighted automata to infinite words (see~\cite{DrosteK03,DrosteGastin07,LatticeAutomata07,Bojanczyk10}
and~\cite{HandbookWA} for a survey), but none of those have considered mean-payoff conditions.
Examples where the mean-payoff measure has been used to specify long-run behaviours of systems 
can be found in game theory~\cite{EM79,ZwickP96} and in Markov decision processes~\cite{Alfaro98}.
The mean-payoff automata as a specification language have been first investigated in~\cite{CDH08,CDH09b,CDH-FCT09}, 
and extended in~\cite{AlurDMW09} to construct a new class of (non-quantitative) languages of infinite words 
(the multi-threshold mean-payoff languages), obtained by applying a query to a 
mean-payoff language, and for which emptiness is decidable. 
It turns out that a richer language of queries can be expressed 
using mean-payoff automaton expressions (together with decidability of the
emptiness problem). A detailed comparison with the results of~\cite{AlurDMW09}
is given in Section~\ref{sec:decidable}. Moreover, we provide algorithmic solutions
to the quantitative language inclusion and equivalence problems and to distance computation which have no
counterpart for non-quantitative languages. Related notions of metrics have been
addressed in stochastic games~\cite{AlfaroMRS07} and probabilistic processes~\cite{DesharnaisGJP99,VidalTHCC05}.


\section{Mean-Payoff Automaton Expressions}\label{sec:mpa-expressions}

\noindent{\bf Quantitative languages.}
A \emph{quantitative language} $L$ over a finite alphabet $\Sigma$ is a function 
$L: \Sigma^{\omega} \to \real$.
Given two quantitative languages $L_1$ and $L_2$ over $\Sigma$, we denote by $\max(L_1,L_2)$ 
(resp., $\min(L_1,L_2)$, $\Sum(L_1,L_2)$ and $-L_1$)
the quantitative language that assigns $\max(L_1(w),L_2(w))$
(resp., $\min(L_1(w), L_2(w))$, $L_1(w) + L_2(w)$, and $-L_1(w)$) to each word 
$w \in \Sigma^{\omega}$. 
The quantitative language $-L$ is called the \emph{complement} of $L$.
The $\max$ and $\min$ operators for quantitative languages
correspond respectively to the least upper bound and 
greatest lower bound for the pointwise order $\preceq$
such that $L_1 \preceq L_2$ if  $L_1(w) \leq L_2(w)$ for all $w \in \Sigma^{\omega}$.
Thus, they generalize respectively the union and intersection operators 
for classical boolean languages. 

\smallskip\noindent{\bf Weighted automata.}
A \emph{$\rat$-weighted automaton} is a tuple $A=\tuple{Q,q_I,\Sigma,\delta,\weight}$, where
\begin{itemize}
\item $Q$ is a finite set of states, $q_I \in Q$ is the initial state, and $\Sigma$ is a finite alphabet;
\item $\delta \subseteq Q \times \Sigma \times Q$ is a finite set of labelled transitions.
We assume that $\delta$ is \emph{total}, i.e., for all 
$q \in Q$ and $\sigma \in \Sigma$, there exists $q'$ such that 
$(q,\sigma,q') \in \delta$; 
\item $\weight: \delta \to \rat$ is a \emph{weight} function, where $\rat$ is the 
set of rational numbers. We assume that rational numbers are encoded as pairs of 
integers in binary.
\end{itemize}
We say that $A$ is \emph{deterministic} if for all 
$q \in Q$ and $\sigma \in \Sigma$, there exists $(q,\sigma,q') \in \delta$ 
for exactly one $q' \in Q$. We sometimes call automata \emph{nondeterministic}
to emphasize that they are not necessarily deterministic. 


\smallskip\noindent{\bf Words and runs.} A \emph{word} $w \in \Sigma^\omega$ is an 
infinite sequence of letters from $\Sigma$.
A \emph{lasso-word} $w$ in $\Sigma^{\omega}$ is an ultimately periodic word of the form 
$w_1 \cdot w_2^{\omega}$ where $w_1 \in \Sigma^*$ is a finite prefix, and $w_2 \in \Sigma^+$ is nonempty.
A \emph{run} of $A$ over an infinite word $w=\sigma_1 \sigma_2 \dots$ 
is an infinite sequence $r = q_0 \sigma_1 q_1 \sigma_2 \dots $ 
of states and letters such that 
\begin{compressEnum}
\itCompress $q_0 = q_I$, and
\itCompress $(q_i,\sigma_{i+1},q_{i+1}) \in \delta$ for all $i \geq 0$.
\end{compressEnum}
We denote by $\weight(r) = v_0 v_1 \dots$ the sequence of weights that occur in~$r$ 
where $v_i = \weight(q_i,\sigma_{i+1},q_{i+1})$ for all $i \geq 0$.

\smallskip\noindent{\bf Quantitative language of mean-payoff automata.}    
The \emph{mean-payoff value} (or limit-average) of a sequence $\bar{v} = v_0 v_1 \dots$ of real numbers
is either 
$$\LimInfAvg(\bar{v}) = \liminf\limits_{n \to \infty} \frac{1}{n} \cdot\sum_{i=0}^{n-1} v_i,
\text{ or}\quad \LimSupAvg(\bar{v}) = \limsup\limits_{n \to \infty} \frac{1}{n}\cdot \sum_{i=0}^{n-1} v_i.$$
Note that if we delete or insert finitely many values in an infinite sequence of numbers,
its limit-average does not change, and if the sequence is ultimately periodic, then
the $\LimInfAvg$ and $\LimSupAvg$ values coincide (and correspond to the mean of the weights
on the periodic part of the sequence).

For $\Val \in \{\LimInfAvg, \LimSupAvg\}$, the quantitative language $L_A$ of $A$ is defined by 
$L_A(w) = \sup \{\Val(\weight(r)) \mid r \text{ is a run of } A \text{ over } w\}$
for all $w \in \Sigma^{\omega}$. 
Accordingly, the automaton $A$ and its quantitative language $L_A$
are called $\LimInfAvg$ or $\LimSupAvg$. 
Note that for deterministic automata, we have $L_A(w) = \Val(\weight(r))$ where 
$r$ is the unique run of $A$ over $w$.

We omit the weight function $\weight$ when it is
clear from the context, and we write $\LimAvg$ when the value according to $\LimInfAvg$ and $\LimSupAvg$
coincide (e.g., for runs with a lasso shape). 

\smallskip\noindent{\bf Decision problems and distance.}
We consider the following classical decision problems for quantitative languages,
assuming an effective presentation of quantitative languages (such as mean-payoff automata,
or automaton expressions defined later).
Given a quantitative language $L$ and a threshold $\nu \in \rat$, 
the \emph{quantitative emptiness problem} 
asks whether there exists a word $w \in \Sigma^{\omega}$
such that $L(w) \geq \nu$, and the \emph{quantitative universality problem}
asks whether $L(w) \geq \nu$ for all words $w \in \Sigma^{\omega}$.

Given two quantitative languages $L_1$ and $L_2$, 
the \emph{quantitative language-inclusion problem} asks whether $L_{1}(w) \leq L_{2}(w)$
for all words $w \in \Sigma^{\omega}$,
and the \emph{quantitative language-equivalence problem} asks whether $L_{1}(w) = L_{2}(w)$
for all words $w \in \Sigma^{\omega}$. Note that universality is a special
case of language inclusion where $L_1$ is constant. Finally, the \emph{distance}
between $L_1$ and $L_2$ is $D_{\sup}(L_1,L_2) = \sup_{w \in \Sigma^{\omega}} \abs{L_1(w) - L_2(w)}$.
It measures how close is an implementation $L_1$ as compared to a specification $L_2$.

It is known that quantitative emptiness is decidable for nondeterministic mean-payoff
automata~\cite{CDH08}, while decidability was open for alternating mean-payoff
automata, as well as for the quantitative language-inclusion problem of
nondeterministic mean-payoff automata. Recent undecidability results on
games with imperfect information and mean-payoff objective~\cite{DDGRT10} entail that 
these problems are undecidable (see Theorem~\ref{thrm_undecidable}). 

\smallskip\noindent{\bf Robust quantitative languages.} 
A class ${\cal Q}$ of quantitative languages is \emph{robust} if the class is 
closed under $\max,\min,\Sum$ and complementation operations. The closure 
properties allow quantitative languages from a robust class to be described 
compositionally.
While nondeterministic $\LimInfAvg$- and $\LimSupAvg$-automata are closed
under the $\max$ operation, they are not closed under $\min$ and complement~\cite{CDH09b}.
Alternating $\LimInfAvg$- and $\LimSupAvg$-automata\footnote{See~\cite{CDH-FCT09} for 
the definition of alternating $\LimInfAvg$- and $\LimSupAvg$-automata that generalize
nondeterministic automata.} 
are closed under $\max$ and $\min$, but are not 
closed under complementation and $\Sum$~\cite{CDH-FCT09} 
We define a \emph{robust} class of quantitative languages for mean-payoff automata which is closed 
under $\max$, $\min$, $\Sum$, and complement,
and which can express all natural examples of quantitative languages defined using the mean-payoff 
measure~\cite{AlurDMW09,CDH09b,CGHIKPS08}.

\smallskip\noindent{\bf Mean-payoff automaton expressions.}
A \emph{mean-payoff automaton expression} $E$ is obtained by the following grammar rule:
$$ E ::=  A \mid \max(E,E) \mid \min(E,E) \mid \Sum(E, E)$$
where $A$ is a \emph{deterministic} $\LimInfAvg$- or $\LimSupAvg$-automaton.
The quantitative language $L_E$ of a mean-payoff automaton expression $E$ 
is~$L_E = L_A$ if $E = A$ is a deterministic automaton,
and $L_E = {\rm op}( L_{E_1}, L_{E_2} )$ if  $E = {\rm op}(E_1,E_2)$
for ${\rm op} \in \{\max,\min,\Sum\}$.
By definition, the class of mean-payoff automaton expression is 
closed under $\max$, $\min$ and $\Sum$.
Closure under complement follows from the fact that the complement 
of $\max(E_1,E_2)$ is $\min(-E_1,-E_2)$, the complement of 
$\min(E_1,E_2)$ is $\max(-E_1,-E_2)$, the complement of $\Sum(E_1,E_2)$ 
is $\Sum(-E_1,-E_2)$, and the complement of a deterministic
$\LimInfAvg$-automaton can be defined by the same automaton with 
opposite weights and interpreted as a $\LimSupAvg$-automaton, and 
vice versa, 
since $- \limsup(v_0, v_1, \dots) = \liminf(-v_0, \break -v_1, \dots)$.
Note that arbitrary linear combinations of deterministic mean-payoff automaton expressions
(expressions such as $c_1 E_1 + c_2 E_2$ where $c_1,c_2 \in \rat$ 
are rational constants) can be obtained for free since scaling the weights of a 
mean-payoff automaton by a positive factor $\abs{c}$ results in a 
quantitative language scaled by the same factor.

\section{The Vector Set of Mean-Payoff Automaton Expressions}\label{sec:vector-set}

Given a mean-payoff automaton expression $E$, let $A_1, \dots, A_n$ be the deterministic weighted 
automata occurring in $E$.
The \emph{vector set} of $E$ is the set 
$V_E = \{\tuple{L_{A_1}(w),\ldots,L_{A_n}(w)} \in \real^n \mid w \in \Sigma^{\omega}\}$
of tuples of values of words according to each automaton $A_i$.
In this section, we characterize the vector set of mean-payoff automaton expressions,
and in Section~\ref{sec:alg-cons} we give an algorithmic procedure to compute this set.
This will be useful to establish the decidability of all decision problems, and
to compute the distance between mean-payoff automaton expressions.
Given a vector $v \in \real^n$, we denote by $\norm{v} = \max_i \,\abs{v_i}$ the \emph{$\infty$-norm} of $v$.

The \emph{synchronized product} of $A_1, \dots, A_n$ such that $A_i = \tuple{Q_i,q^i_I,\Sigma,\delta_i,\weight_i}$
is the $\rat^n$-weighted automaton $A_E = A_1 \times \dots \times A_n = \tuple{Q_1 \times \dots \times Q_n,
(q^1_I, \dots, q^n_I), \Sigma,\delta,\weight}$ such that 
$t = ((q_1, \dots, q_n), \sigma, (q'_1, \dots, q'_n)) \in \delta$ if $t_i:=(q_i, \sigma, q'_i) \in \delta_i$
for all \mbox{$1 \leq i \leq n$}, and $\weight(t) = (\weight_1(t_1), \dots, \weight_n(t_n))$.
In the sequel, we assume that all $A_i$'s are deterministic $\LimInfAvg$-automata (hence, $A_E$ is deterministic) 
and that the underlying graph of the automaton $A_E$ 
has only one strongly connected component (scc). 
We show later how to obtain the vector set without these restrictions.

For each (simple) cycle $\rho$ in $A_E$, let the \emph{vector value} of $\rho$ be the mean of the 
tuples labelling the edges of $\rho$, denoted $\Avg(\rho)$. To each simple cycle $\rho$ in $A_E$ corresponds a (not necessarily simple)
cycle in each $A_i$, and the vector value $(v_1,\dots, v_n)$ of $\rho$ contains the mean value $v_i$ of $\rho$
in each $A_i$. We denote by $S_E$ the (finite) set of vector values of simple cycles in $A_E$.
Let $\conv(S_E)$ be the convex hull of $S_E$.

\begin{lemma}\label{lem:convex-hull-lasso-words}
Let $E$ be a mean-payoff automaton expression. The set
$\conv(S_E)$ is the closure of the set $\{L_E(w) \mid w \text{ is a lasso-word} \}$. 
\end{lemma}

The vector set of $E$ contains more values than the convex hull $\conv(S_E)$,
as shown by the following example. 

\begin{figure}[!t]
\hrule
 \begin{center}
   \unitlength=.8mm
\def\fsize{\normalsize}

\begin{picture}(108,35)(0,0)

{\fsize

\node[Nmarks=i](x0)(12,10){$q_1$}
\node[Nmarks=N, ExtNL=y, NLangle=-90, NLdist=3](x0)(12,10){$A_1$}

\node[Nmarks=i](x1)(42,10){$q_2$}
\node[Nmarks=N, ExtNL=y, NLangle=-90, NLdist=3](x1)(42,10){$A_2$}

\drawloop[ELside=l, ELdist=1, loopCW=y, loopdiam=7, loopangle=90](x0){$\begin{array}{l}a,1 \\ b,0\end{array}$}
\drawloop[ELside=l, ELdist=1, loopCW=y, loopdiam=7, loopangle=90](x1){$\begin{array}{l}a,0 \\ b,1\end{array}$}


\gasset{Nh=1,Nw=1,Nmr=.5,fillgray=0, NLdist=2}

\node[Nframe=n, fillgray=0.6, ExtNL=y, NLangle=180](x00)(70,1){$(0,0)$}
\node[Nmarks=n, ExtNL=y](x01)(70,26){$(0,1)$}
\node[Nmarks=n, ExtNL=y, NLangle=0](x10)(95,1){$(1,0)$}
\drawedge[ELpos=55, ELside=l, ELdist=1, AHnb=0, curvedepth=0](x01,x10){$H=\conv(S_E)$}

\drawedge[ELpos=55, ELside=l, ELdist=1, AHnb=0, curvedepth=0, linegray=0.6, dash={1}0](x01,x00){}
\drawedge[ELpos=40, ELside=l, ELdist=1, AHnb=0, curvedepth=0, linegray=0.6, dash={1}0](x00,x10){$F_{\min}(H)$}

}
\end{picture}
 \end{center}
\hrule
 \caption{The vector set of $E = \max(A_1,A_2)$ is $F_{\min}(\conv(S_E)) \supsetneq \conv(S_E)$. \label{fig:value-set}}
\end{figure}
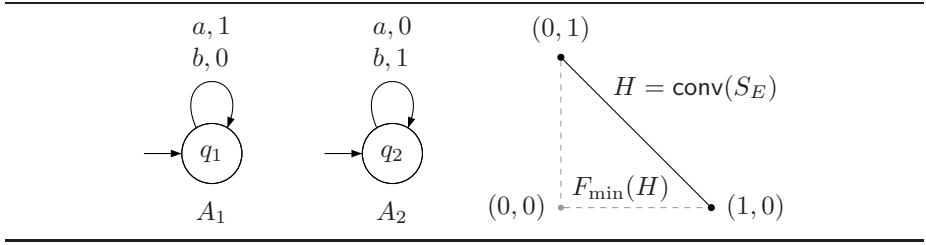

\begin{example}\label{ex1}
Consider the expression $E = \max(A_1,A_2)$ where $A_1$ and $A_2$ are 
deterministic $\LimInfAvg$-automata (see \figurename~\ref{fig:value-set}).
The product $A_E = A_1 \times A_2$ has two simple cycles with respective vector values
$(1,0)$ (on letter `$a$') and $(0,1)$ (on letter `$b$'). The set $H=\conv(S_E)$
is the solid segment on \figurename~\ref{fig:value-set} and contains the vector
values of all lasso-words. However, other vector values can be obtained: consider
the word $w = a^{n_1} b^{n_2} a^{n_3} b^{n_4} \dots$ where $n_1 = 1$ and $n_{i+1} = (n_1 + \cdots + n_i)^2$
for all $i \geq 1$. It is easy to see that the value of $w$ according to $A_1$ is $0$
because the average number of $a$'s in the prefixes $a^{n_1} b^{n_2}  \dots a^{n_i} b^{n_{i+1}}$
for $i$ odd is smaller than $\frac{n_1 + \cdots + n_i}{n_1 + \cdots + n_i + n_{i+1}} = \frac{1}{1+ n_1 + \cdots + n_i}$
which tends to $0$ when $i \to \infty$. Since $A_1$ is a $\LimInfAvg$-automaton,
the value of $w$ is $0$ in $A_1$, and by a symmetric argument the value of $w$ is also 
$0$ in $A_2$. Therefore the vector $(0,0)$ is in the vector set of $E$. 
Note that $z=(0,0)$ is the pointwise minimum of $x=(1,0)$ and $y=(0,1)$, i.e. $(0,0) = f_{\min}((1,0), (0,1))$
where $z = f_{\min}(x,y)$ if $z_1 = \min(x_1,y_1)$ and $z_2 = \min(y_1,y_2)$.
In fact, the vector set is the whole triangular region in \figurename~\ref{fig:value-set},
i.e. $V_E = \{ f_{\min}(x,y) \mid x,y \in \conv(S_E) \}$.
\end{example}

We generalize $f_{\min}$ to finite sets of points $P \subseteq \real^n$ in $n$ dimensions as follows: $f_{\min}(P) \in \real^n$
is the point $p = (p_1, p_2, \ldots, p_n)$ such that $p_i$ is the minimum $i^{\text{th}}$ coordinate of the points in $P$,
for $1 \leq i \leq n$. For arbitrary $S \subseteq \real^n$, define $F_{\min}(S) = \{f_{\min}(P) \mid P \text{ is a finite subset of } S\}$.
As illustrated in Example~\ref{ex1}, the next lemma shows that the vector set $V_E$ is equal to $F_{\min}(\conv(S_E))$.





\begin{lemma}\label{lem:value-set}
Let $E$ be a mean-payoff automaton expression built from deterministic $\LimInfAvg$-automata, and such that
$A_E$ has only one strongly connected component. 
Then, the vector set of $E$ is $V_E = F_{\min}(\conv(S_E))$.
\end{lemma}

For a general mean-payoff automaton expression~$E$ (with both deterministic $\LimInfAvg$- and $\LimSupAvg$ automata, and
with multi-scc underlying graph), we can use the result of Lemma~\ref{lem:value-set} as follows.
We replace each $\LimSupAvg$ automaton $A_i$ occurring in $E$ by the $\LimInfAvg$ automaton $A'_i$ obtained
from $A_i$ by replacing every weight $\weight$ by $-\weight$. The duality of $\liminf$ and $\limsup$ yields
$L_{A'_i} = - L_{A_i}$. In each strongly connected component $\C$ of the underlying graph of $A_E$, we compute 
$V_{\C} = F_{\min}(\conv(S_{\C}))$ (where $S_{\C}$ is the set of vector values of the simple cycles in $\C$) and apply 
the transformation $x_i \to -x_i$ on every coordinate~$i$ where the automaton
$A_i$ was originally a $\LimSupAvg$ automaton. The union of the sets $\bigcup_{\C} V_{\C}$ where $\C$ ranges over
the strongly connected components of $A_E$ gives the vector set of $E$.

\begin{theorem}\label{theo:value-set-general}
Let $E$ be a mean-payoff automaton expression built from deterministic $\LimInfAvg$-automata, and let
$\Z$ be the set of strongly connected components in $A_E$. 
For a strongly connected component $\C$ let $S_\C$ denote the set of vector values of the simple cycles 
in $\C$.
The vector set of $E$ is $V_E = \bigcup_{\C \in \Z} F_{\min}(\conv(S_\C))$.
\end{theorem}

\section{Computation of $F_{\min}(\conv(S))$ for a Finite Set $S$}\label{sec:alg-cons}
It follows from Theorem~\ref{theo:value-set-general} that the vector set $V_E$ of a mean-payoff automaton expression~$E$ can
be obtained as a union of sets $F_{\min}(\conv(S))$, where $S \subseteq \real^n$ is a finite set. 
However, the set $\conv(S)$ being in general infinite, it is not immediate that 
$F_{\min}(\conv(S))$ is computable. 
In this section we consider the problem of computing $F_{\min}(\conv(S))$ for a finite set $S$.
In subsection~\ref{subsec:explicit} we present an explicit construction and 
in subsection~\ref{subsec:geometric} we give a geometric construction of 
the set as a set of linear constraints.
We first present some properties of the set $F_{\min}(\conv(S))$.

\begin{lemma}\label{lem:value-set-convex}
If $X$ is a convex set, then $F_{\min}(X)$ is convex.
\end{lemma}

By Lemma~\ref{lem:value-set-convex}, the set $F_{\min}(\conv(S))$ is convex, and since 
$F_{\min}$ is a monotone operator and $S \subseteq \conv(S)$, we have  $F_{\min}(S) \subseteq F_{\min}(\conv(S))$
and thus $\conv(F_{\min}(S)) \subseteq F_{\min}(\conv(S))$.
The following proposition states that in two dimensions the above sets coincide.

\begin{proposition}\label{prop:conv-min-2D}
Let $S \subseteq \real^2$ be a finite set.
Then, $\conv(F_{\min}(S)) = F_{\min}(\conv(S))$.
\end{proposition}

%

We show in the following example that in three dimensions the above proposition does not hold, i.e., we show 
that $F_{\min}(\conv(S_E)) \neq \conv(F_{\min}(S_E))$ in $\real^3$.

\begin{example}\label{examp2}
We show that in three dimension there is a finite set $S$ such that  
$F_{\min}(\conv(S)) \not\subseteq \conv(F_{\min}(S))$.
Let $S = \{q,r,s\}$ with $q=(0,1,0)$, $r=(-1,-1,1)$, and $s=(1,1,1)$.
Then $f_{\min}(r,s) = r$, $f_{\min}(q,r,s) = f_{\min}(q,r) = t = (-1,-1,0)$, and $f_{\min}(q,s) = q$.
Therefore $F_{\min}(S) = \{q,r,s,t\}$. Consider $p = (r+s)/2 = (0,0,1)$. We have $p \in \conv(S)$
and $f_{\min}(p,q) = (0,0,0)$. Hence $(0,0,0) \in F_{\min}(\conv(S))$.
We now show that $(0,0,0)$ does not belong to $\conv(F_{\min}(S))$.
Consider $u= \alpha_q \cdot q + \alpha_r \cdot r + \alpha_s \cdot s + \alpha_t \cdot t$ 
such that $u$ in $\conv(F_{\min}(S))$.
Since the third coordinate is non-negative for $q,r,s$, and $t$, it follows that if $\alpha_r>0$ 
or $\alpha_s>0$, then the third coordinate of $u$ is positive. 
If $\alpha_s=0$ and $\alpha_r=0$, then we have two cases:
(a) if $\alpha_t>0$, then the first coordinate of $u$ is negative; and
(b) if $\alpha_t=0$, then the second coordinate of $u$ is~1.
It follows $(0,0,0)$ is not in $\conv(F_{\min}(S))$.
\qed
\end{example}

\subsection{Explicit construction}\label{subsec:explicit}
Example~\ref{examp2} shows that in general $F_{\min}(\conv(S)) \not\subseteq \conv(F_{\min}(S))$. 
In this section we present an explicit construction that given a finite set $S$ constructs a 
finite set $S'$ such that (a)~$S \subseteq S' \subseteq \conv(S)$ and (b)~$F_{\min}(\conv(S)) \subseteq \conv(F_{\min}(S'))$.
It would follow that $F_{\min}(\conv(S))=\conv(F_{\min}(S'))$. 
Since convex hull of a finite set is computable and $F_{\min}(S')$ is finite, 
this would give us an algorithm to compute $F_{\min}(\conv(S))$.
For simplicity, for the rest of the section we write $F$ for $F_{\min}$ and 
$f$ for $f_{\min}$ (i.e., we drop the $\min$ from subscript).
Recall that $F(S)=\set{f(P) \mid  P \text{ finite subset of } S}$ 
and let $F_i(S)=\set{f(P) \mid  P \text{ finite subset of $S$ and } |P| \leq i }$.
We consider $S \subseteq \real^n$.

\begin{lemma}\label{lemm1}
Let $S \subseteq \real^n$. Then, $F(S)=F_n(S)$ and $F_n(S) \subseteq F_2^{n-1}(S)$.
\end{lemma}

\smallskip\noindent{\bf Iteration of a construction $\constr$.} 
We will present a 
construction $\constr$ with the following properties: input to the construction
is a finite set $Y$ of points, and the output $\constr(Y)$ satisfies the 
following properties
\begin{enumerate}
\item {\bf (Condition C1).} $\constr(Y)$ is finite and subset of $\convk(Y)$. 
\item {\bf (Condition C2).} $F_2(\convk(Y)) \subseteq \convk(F (\constr(Y)))$.
\end{enumerate}
Before presenting the construction $\constr$ we first show how to iterate the
construction to obtain the following result: given a finite set of points $X$
we construct a finite set of points $X'$ such that 
$F(\convk (X)) = \convk(F (X'))$.

\smallskip\noindent{\em Iterating $\constr$.} 
Consider a finite set of points $X$, 
and let $X_0=X$ and $X_1=\constr(X_0)$. Then we have 
\[
\convk(X_1) \subseteq \convk (\convk (X_0)) \quad 
\text{(since by Condition {\bf C1} we have } 
X_1 \subseteq \convk (X_0))
\]
and hence $\convk(X_1) \subseteq \convk(X_0)$; and 
\[
F_2(\convk(X_0)) \subseteq \convk(F(X_1))
\qquad 
\text{(by Condition {\bf C2})} 
\]
By iteration we obtain that $X_i=\constr(X_{i-1})$ for $i \geq 2$ and as above
we have 
\[
(1)\ \convk(X_i) \subseteq \convk(X_0) \qquad 
(2)\ F_2^i(\convk(X_0)) \subseteq \convk(F(X_i))
\] 
Thus for $X_{n-1}$ we have
\[
(1)\ \convk(X_{n-1}) \subseteq \convk(X_0) \qquad 
(2)\ F_2^{n-1}(\convk(X_0)) \subseteq \convk(F(X_{n-1}))
\] 
By (2) above and Lemma~\ref{lemm1}, we obtain 
\[
(A)\ F(\convk (X_0)) = F_n(\convk(X_0)) \subseteq F_2^{n-1}(\convk(X_0)) 
\subseteq \convk(F(X_{n-1}))
\]
By (1) above we have $\convk(X_{n-1}) \subseteq \convk (X_0)$ and hence
$F(\convk(X_{n-1})) \subseteq F(\convk (X_0))$.
Thus we have 
\[
\convk(F(\convk(X_{n-1}))) \subseteq \convk (F(\convk (X_0))) = F(\convk(X_0)) 
\]
where the last equality follows since by Lemma~\ref{lem:value-set-convex} we have 
$F(\convk(X_0))$ is convex.
Since $X_{n-1} \subseteq \convk(X_{n-1})$ we have 
\[
(B)\ 
\convk(F(X_{n-1})) 
\subseteq \convk(F(\convk(X_{n-1}))) \subseteq F(\convk(X_0))
\]
Thus by (A) and (B) above we have 
$F(\convk(X_0)) = \convk(F(X_{n-1}))$.
Thus given the finite set $X$, we have the finite set $X_{n-1}$ such 
that (a)~$X \subseteq X_{n-1} \subseteq \conv(X)$ and (b)~$F(\conv(X)) = \conv(F(X_{n-1}))$.  
We now present the construction $\constr$ to complete the result.

\smallskip\noindent{\bf The construction $\constr$.}
Given a finite set $Y$ of points $Y'=\constr(Y)$ is obtained by adding points 
to $Y$ in the following way:
\begin{itemize}
\item For all $1 \leq k \leq n$, we consider all $k$-dimensional coordinate 
planes $\Pi$ supported by a point in $Y$;
\item Intersect each coordinate plane $\Pi$ with $\convk(Y)$ and the result is 
a convex polytope $Y_\Pi$;
\item We add the corners (or extreme points) of each polytope $Y_\Pi$ to $Y$.
\end{itemize}
The proof that the above construction satisfies condition {\bf C1} and {\bf C2} is
given in the appendix, and thus we have the following result.

\begin{theorem}\label{theo:explicit-construction}
Given a finite set $S \subseteq \real^n$ such that $|S|=m$, 
the following assertion holds: 
a finite set $S'$ with $|S'| \leq m^{2^n} \cdot 2^{n^2+n}$ can be 
computed in $m^{O(n \cdot 2^n)} \cdot 2^{O(n^3)}$ time such that 
(a)~$S \subseteq S' \subseteq \conv(S)$ and 
(b)~$F_{\min}(\conv(S))= \conv(F_{\min}(S'))$.
\end{theorem}

\subsection{Linear constraint construction}\label{subsec:geometric}
In the previous section we presented an explicit construction of a finite 
set of points whose convex hull gives us $F_{\min}(\conv(S))$. 
The explicit construction illuminates properties of the set $F_{\min}(\conv(S))$,
however, the construction is inefficient computationally.
In this subsection we present an efficient geometric construction for the 
computation of $F_{\min}(\conv(S))$ for a finite set $S$.
Instead of constructing a finite set $S'\subseteq \conv(S)$ such that 
$\conv(S') =F_{\min}(\conv(S))$, we represent $F_{\min}(\conv(S))$ as a finite set of 
linear constraints.



Consider the \emph{positive orthant} anchored at the origin
in $\real^n$, that is, the set of points with non-negative
coordinates:
  $\real_+^n =  \{ (z_1, z_2, \ldots, z_n) \mid z_i \geq 0, \forall i \}$.
Similarly, the \emph{negative orthant} is the set of points with
non-positive coordinates, denoted as $\real_-^n = - \real_+^n$.
Using vector addition, we write $y + \real_+^n$ for the positive
orthant anchored at $y$.
Similarly, we write $x + \real_-^n = x - \real_+^n$
for the negative orthant anchored at $x$.
The positive and negative orthants satisfy the following simple
\emph{duality relation}:
%
  $x \in y + \real_+^n$ iff $y \in x - \real_+^n$.


Note that $\real_+^n$ is an $n$-dimensional convex polyhedron.
For each $1 \leq j \leq n$, we consider the $(n-1)$-dimensional face $\ratl_j$ spanned
by the coordinate axes except the $j^{\text{th}}$ one, that is,
  $\ratl_j =  \{ (z_1, z_2, \ldots, z_n) \in \real_+^n \mid z_j = 0 \}$.

We say that $y + \real_+^n$ is \emph{supported} by $X$
if $(y + \ratl_j) \capsp X \neq \emptyset$ for every $1 \leq j \leq n$.
Assuming $y + \real_+^n$ is supported by $X$, we can construct a
set $Y \subseteq X$ by collecting one point per
$(n-1)$-dimensional face of the orthant and get $y = f(Y)$.
It is also allowed that two faces contribute the same point to $Y$.
Similarly, if $y = f(Y)$ for a subset $Y \subseteq X$,
then the positive orthant anchored at $y$ is supported by $X$.
Hence, we get the following lemma.

\begin{lemma}[Orthant Lemma]
  $y \in F_{\min}(X)$ iff $y + \real_+^n$ is supported by $X$.
\end{lemma}

\paragraph{Construction.}
We use the Orthant Lemma to construct $F_{\min}(X)$.
We begin by describing the set of points $y$ for which the $j^{\text{th}}$
face of the positive orthant anchored at $y$
has a non-empty intersection with $X$.
Define $F_j = X - \ratl_j$, the set of points of the form $x - z$,
where $x \in X$ and $z \in \ratl_j$.

\begin{lemma}[Face Lemma]
  $(y + \ratl_j) \capsp X \neq \emptyset$ iff $y \in F_j$.
\end{lemma}

\begin{proof}
 Let $x \in X$ be a point in the intersection, that is,
 $x \in y + \ratl_j$.
 Using the duality relation for the $(n-1)$-dimensional orthant,
 we get $y \in x - \ratl_j$.
 By definition, $x - \ratl_j$ is a subset of $X - \ratl_j$,
 and hence $y \in F_j$.
\qed
\end{proof}

It is now easy to describe the set defined in our problem statement.

\begin{lemma}[Characterization]
  $F_{\min}(X) = \bigcap_{j=1}^n F_j$.
\end{lemma}
\begin{proof}
 By the Orthant Lemma, $y \in F_{\min}(X)$ iff $y + \real_+^n$ is supported by $X$.
 Equivalently, $(y + \ratl_j) \capsp X \neq \emptyset$
 for all $1 \leq j \leq n$.
 By the Face Lemma, this is equivalent to $y$ belonging to the common
 intersection of the sets 
 $F_j = X - \ratl_j$.
\qed
\end{proof}

\noindent{\bf Algorithm for computation of $F_{\min}(\conv(S))$.}
Following the construction, we get an algorithm that computes 
$F_{\min}(\conv(S))$ for a finite set $S$ of points in $\real^n$. 
Let $|S|=m$. 
We first represent $X=\conv(S)$ as intersection of half-spaces:
we require at most $m^n$ half-spaces (linear constraints).
It follows that $F_j=X-\ratl_j$ can be expressed as $m^n$ 
linear constraints, and hence $F_{\min}(X)=\bigcap_{j=1}^n F_j$ can be expressed
as $n\cdot m^n$ linear constraints. 
This gives us the following result.



\begin{theorem}
Given a finite set $S$ of $m$ points in $\real^n$, we can construct in 
$O(n \cdot m^n)$ time $n \cdot m^n$ linear constraints that represent 
$F_{\min}(\conv(S))$. 
\end{theorem}

\section{Mean-Payoff Automaton Expressions are Decidable}\label{sec:decidable}

Several problems on quantitative languages can be solved for the class
of mean-payoff automaton expressions using the vector set. 
The decision problems of quantitative emptiness and universality, and quantitative language inclusion and equivalence
are all decidable, as well as questions related to cut-point languages, and computing 
distance between mean-payoff languages. 

\paragraph{Decision problems and distance.}
From the vector set $V_E = \{\tuple{L_{A_1}(w),\ldots,L_{A_n}(w)} \in \real^n \mid w \in \Sigma^{\omega}\}$,
we can compute the \emph{value set} $L_E(\Sigma^{\omega}) = \{L_E(w) \mid w \in \Sigma^{\omega}\}$ of values
of words according to the quantitative language of $E$ as follows. The set $L_E(\Sigma^{\omega})$ is obtained by successive application
of \emph{min-}, \emph{max-} and \emph{sum-projections} $p^{\min}_{ij}, p^{\max}_{ij}, p^{{\Sum}}_{ij}: \real^k \to \real^{k-1}$ where $i < j \leq k$, 
defined by 
$$\begin{array}{ll}
p^{\min}_{ij}((x_1,\dots, x_k)) = & (x_1, \dots, x_{i-1}, \min(x_i,x_j), x_{i+1}, \dots, x_{j-1}, x_{j+1}, \dots x_k), \\[+3pt] 
p^{{\Sum}}_{ij}((x_1,\dots, x_k)) = & (x_1, \dots, x_{i-1}, \quad\! x_i + x_j \quad\!, x_{i+1}, \dots, x_{j-1}, x_{j+1}, \dots x_k), \\ 
\end{array}
$$
and analogously for $p^{\max}_{ij}$. For example, $p^{\max}_{12}(p^{\min}_{23}(V_E))$ gives the set $L_E(\Sigma^{\omega})$ of word values of the mean-payoff automaton expression $E = \max(A_1,\min(A_2,A_3))$.

Assuming a representation of the polytopes of $V_E$ as a boolean combination $\varphi_E$ of linear constraints, 
the projection  $p^{\min}_{ij}(V_E)$ is represented by the formula 
$$\psi = (\exists x_j: \varphi_E \land x_i \leq x_j) \lor (\exists x_i: \varphi_E \land x_j \leq x_i)[x_j \gets x_i]$$
where $[x \gets e]$ is a substitution that replaces every occurrence of $x$ by the expression $e$.
Since linear constraints over the reals admit effective elimination of existential quantification, the formula $\psi$ can be transformed
into an equivalent boolean combination of linear constraints without existential quantification.
The same applies to max- and sum-projections.

Successive applications of min-, max- and sum-projections (following the structure of the 
mean-payoff automaton expression $E$)
gives the value set $L_E(\Sigma^{\omega}) \subseteq \real$ as a boolean combination
of linear constraints, hence it is a union of intervals. From this set, it is easy to decide the 
quantitative emptiness problem and the quantitative universality problem: there exists a word $w \in \Sigma^{\omega}$ such that $L_E(w) \geq \nu$
if and only if $L_E(\Sigma^{\omega}) \cap \, [\nu, +\infty [ \,\neq \emptyset$, and $L_E(w) \geq \nu$ for all words $w \in \Sigma^{\omega}$ 
if and only if $L_E(\Sigma^{\omega}) \,\cap \, ] -\infty, \nu [ \,= \emptyset$.

In the same way, we can decide the quantitative language inclusion problem ``is $L_E(w) \leq L_F(w)$ for all words $w \in \Sigma^{\omega}$ ?'' 
by a reduction to the universality problem for the expression $F-E$ and threshold $0$ 
since mean-payoff automaton expressions are closed under sum and complement. 
The quantitative language equivalence problem is then obviously also decidable.

Finally, the distance between the quantitative languages of $E$ and $F$
can be computed as the largest number (in absolute value) in the value set of $F-E$.
As a corollary, this distance is always a rational number.


\paragraph{Comparison with~\cite{AlurDMW09}.}
The work in~\cite{AlurDMW09} considers deterministic mean-payoff automata with multiple payoffs.
The weight function in such an automaton is of the form $\weight: \delta \to \rat^d$. The value of a finite
sequence $(v_i)_{1 \leq i \leq n}$ (where $v_i \in \rat^d$) is the mean of the tuples $v_i$, that is a $d$-dimensional vector $\Avg_n = \frac{1}{n} \cdot \sum_{i=0}^{n-1} v_i$.
The ``value'' associated to an infinite run (and thus also to the corresponding word, since the automaton is deterministic) is the
set $Acc \subseteq \real^d$ of accumulation points of the sequence $(\Avg_n)_{n \geq 1}$. 


In~\cite{AlurDMW09}, a query language on the set of accumulation points is used to define \emph{multi-threshold mean-payoff languages}. 
For $1 \leq i \leq n$, let $p_i: \real^n \to \real$ be the usual projection along the $i^{\text{th}}$ coordinate.
A query is a boolean combination of atomic threshold conditions of the form $\min(p_i(Acc)) \sim \nu$ or $\max(p_i(Acc)) \sim \nu$ where $\sim \in \{<,\leq,\geq,>\}$
and $\nu \in \rat$. A word is accepted if the set of accumulation points of its (unique) run satisfies the query. 
Emptiness is decidable for such multi-threshold mean-payoff languages, by an argument based on the computation
of the convex hull of the vector values of the simple cycles in the automaton~\cite{AlurDMW09} (see also Lemma~\ref{lem:convex-hull-lasso-words}). 
We have shown that this convex hull $\conv(S_E)$ is not sufficient to analyze quantitative languages of mean-payoff automaton expressions. 
It turns out that a richer query language can also be defined using our construction of $F_{\min}(\conv(S_E))$.

In our setting, we can view a $d$-dimensional mean-payoff automaton $A$ as a product $P_A$ of 2d copies $A^i_{t}$ of $A$ (where $1 \leq i \leq d$ and $t \in \{\LimInfAvg, \LimSupAvg\}$), 
where $A^i_{t}$ assigns to each transition the $i^{\text{th}}$  coordinate of the payoff vector in $A$, and the automaton is interpreted as a $t$-automaton. 
Intuitively, the set $Acc$ of accumulation points of a word $w$ satisfies $\min(p_i(Acc)) \sim \nu$ (resp. $\max(p_i(Acc) \sim \nu$) if and only if the value of $w$ according to the
automaton $A^i_{t}$ for $t = \LimInfAvg$ (resp. $t = \LimSupAvg$) is $\sim \nu$. Therefore, atomic threshold conditions can be encoded
as threshold conditions on single variables of the vector set for $P_A$. 
Therefore, the vector set computed in Section~\ref{sec:alg-cons} allows to decide
the emptiness problem for multi-threshold mean-payoff languages, by checking emptiness of the intersection 
of the vector set with the constraint corresponding to the query. 

Furthermore, we can solve more expressive queries in our framework, namely where atomic conditions are
linear constraints on $\LimInfAvg$- and $\LimSupAvg$-values.
For example, the constraint $\LimInfAvg(\weight_1) + \LimSupAvg(\weight_2) \sim \nu$
is simply encoded as $x_k + x_l \sim \nu$ where $k,l$ are the indices corresponding to $A^1_{\LimInfAvg}$ and $A^2_{\LimSupAvg}$
respectively. 
Note that the trick of extending the dimension of the $d$-payoff vector with, say $\weight_{d+1} = \weight_1 + \weight_2$,
is not equivalent because 
${\sf Lim}\raisebox{2.0pt}{\scalebox{0.45}{\Big\{\begin{tabular}{c}{\sf Sup}\\[-3pt]{\sf Inf}\end{tabular}\Big\}}}{\sf Avg}(\weight_1) \pm
{\sf Lim}\raisebox{2.0pt}{\scalebox{0.45}{\Big\{\begin{tabular}{c}{\sf Sup}\\[-3pt]{\sf Inf}\end{tabular}\Big\}}}{\sf Avg}(\weight_2)$
is not equal to ${\sf Lim}\raisebox{2.0pt}{\scalebox{0.45}{\Big\{\begin{tabular}{c}{\sf Sup}\\[-3pt]{\sf Inf}\end{tabular}\Big\}}}{\sf Avg}(\weight_1 \pm \weight_2)$
in general (no matter the choice of $\raisebox{2.0pt}{\scalebox{0.45}{\Big\{\begin{tabular}{c}{\sf Sup}\\[-3pt]{\sf Inf}\end{tabular}\Big\}}}$ and $\pm$).
Hence, in the context of non-quantitative languages our results also provide a richer query language for the deterministic mean-payoff automata with multiple payoffs.

\paragraph{Complexity.}
All problems studied in this section can be solved easily (in polynomial time) once 
the value set is constructed, which can be done in quadruple exponential time.
The quadruple exponential blow-up is caused by $(a)$ the synchronized product 
construction for $E$, $(b)$ the computation of the vector values of all simple cycles in $A_E$,
$(c)$ the construction of the vector set $F_{\min}(\conv(S_E))$, and $(d)$
the successive projections of the vector set to obtain the value set.
Therefore, all the above problems can be solved in 4EXPTIME. 

\begin{theorem}\label{thrm_decidable}
For the class of mean-payoff automaton expressions, the quantitative emptiness,
universality, language inclusion, and equivalence problems, as well as distance
computation can be solved in 4EXPTIME.
\end{theorem}

Theorem~\ref{thrm_decidable} is in sharp contrast with the nondeterministic 
and alternating mean-payoff automata for which language inclusion is 
undecidable (see also Table~\ref{tab:properties}). 
The following theorem presents the undecidability result that is derived from 
the results of~\cite{DDGRT10}. 

\begin{theorem}\label{thrm_undecidable}
The quantitative universality, 
language inclusion, and 
language equivalence problems are undecidable 
for nondeterministic mean-payoff automata; and 
the quantitative emptiness, 
universality, 
language inclusion, and 
language equivalence problems are undecidable for 
alternating mean-payoff automata.
\end{theorem}


\section{Expressive Power and Cut-point Languages}\label{sec:expressive-power}

We study the expressive power of mean-payoff automaton expressions $(i)$ according
to the class of quantitative languages that they define, and  $(ii)$ according
to their cut-point languages.

\paragraph{Expressive power comparison.}
We compare the expressive power of 
mean-payoff automaton expressions with nondeterministic and alternating mean-payoff automata.
The results of~\cite{CDH09b} show that there exist deterministic mean-payoff automata 
$A_1$ and $A_2$ such that $\min(A_1,A_2)$ cannot be expressed by nondeterministic mean-payoff 
automata.
The results of~\cite{CDH-FCT09} shows that there exists deterministic mean-payoff 
automata $A_1$ and $A_2$ such that $\Sum(A_1,A_2)$ cannot be expressed by alternating mean-payoff 
automata.
It follows that there exist languages expressible by mean-payoff automaton expression that  
cannot be expressed by nondeterministic and alternating mean-payoff automata.
In Theorem~\ref{thrm_expressive_power} we show the converse, that is, 
we show that there exist languages expressible by nondeterministic 
mean-payoff automata that cannot be expressed by mean-payoff automaton expression.
It may be noted that the subclass of mean-payoff automaton expressions that 
only uses min and max operators (and no sum operator) is a strict subclass of 
alternating mean-payoff automata, and when only the max operator is used we get
a strict subclass of the nondeterministic mean-payoff automata.


\begin{theorem}\label{thrm_expressive_power}
Mean-payoff automaton expressions are incomparable in expressive power with 
nondeterministic and alternating mean-payoff automata: (a)~there exists a quantitative 
language that is expressible by mean-payoff automaton expressions, but cannot be expressed by 
alternating mean-payoff automata; and 
(b)~there exists a quantitative language that is expressible by a nondeterministic 
mean-payoff automaton, but cannot be expressed by a mean-payoff automaton expression.
\end{theorem}

\paragraph{Cut-point languages.}
Let $L$ be a quantitative language over $\Sigma$. 
Given a threshold $\eta \in \real$, the \emph{cut-point language} defined by $(L,\eta)$
is the language (i.e., the set of words) $L^{\geq \eta} = \{w \in \Sigma^{\omega} \mid L(w) \geq \eta \}$.
It is known for deterministic mean-payoff automata that the cut-point language
may not be $\omega$-regular, while it is $\omega$-regular if the threshold $\eta$ 
is \emph{isolated}, i.e. if there exists $\epsilon > 0$ such that $\abs{L(w) - \eta} > \epsilon$ for
all words $w \in \Sigma^{\omega}$~\cite{CDH09b}. 

We present the following results
about cut-point languages of mean-payoff automaton expressions.
First, we note that it is decidable whether a rational threshold $\eta$ is 
an isolated cut-point of a mean-payoff automaton expression, using the value set
(it suffices to check that $\eta$ is not in the value set since this set is closed).
Second, isolated cut-point languages of mean-payoff automaton expressions are \emph{robust}
as they remain unchanged under sufficiently small perturbations of the transition
weights. This result follows from a more general robustness property of weighted
automata~\cite{CDH09b} that extends to mean-payoff automaton expressions: if the
weights in the automata occurring in $E$ are changed by at most $\epsilon$, 
then the value of every word changes by at most $\max(k,1) \cdot \epsilon$ where
$k$ is the number of occurrences of the $\Sum$ operator in $E$.
Therefore $D_{\sup}(L_E,L_{F^\epsilon}) \to 0$ when $\epsilon \to 0$ where  $F^\epsilon$ is any mean-payoff automaton expression
obtained from $E$ by changing the weights by at most $\epsilon$. 
As a consequence, isolated cut-point languages of mean-payoff automaton expressions are robust.
Third, the isolated cut-point language of mean-payoff automaton expressions is 
$\omega$-regular. To see this, note that every strongly connected component
of the product automaton $A_E$ contributes with a closed convex set to the value 
set of $E$. Since the $\max$-, $\min$- and $\Sum$-projections are continuous
functions, they preserve connectedness of sets and therefore each scc $C$ contributes with
an interval $[m_C,M_C]$ to the value set of $E$. An isolated cut-point $\eta$ cannot belong to any of these
intervals, and therefore we obtain a B\"uchi-automaton for the cut-point language
by declaring to be accepting the states of the product automaton $A_E$ that belong
to an scc $C$ such that $m_C > \eta$. Hence, we get the following result.

\begin{theorem}\label{theo:cut-point}
Let $L$ be the quantitative language of a mean-payoff automaton expression.
If $\eta$ is an isolated cut-point of $L$, then the cut-point language $L^{\geq \eta}$
is $\omega$-regular.
\end{theorem}

\vspace{-1em}
\section{Conclusion and Future Works}

We have presented a new class of quantitative languages, the \emph{mean-payoff automaton expressions}
which are both robust and decidable (see Table~\ref{tab:properties}), and for which 
the distance between quantitative languages can be computed.
The decidability results come with a high worst-case complexity, and it is
a natural question for future works to either improve the algorithmic solution,
or present a matching lower bound. 
Another question of interest is to find a robust and decidable class of 
quantitative languages based on the discounted sum measure~\cite{CDH08}.

\bibliography{biblio}
\bibliographystyle{plain}

\newpage
\appendix

\section{Proofs of Section~\ref{sec:vector-set}}

\begin{proof}[of Lemma~\ref{lem:convex-hull-lasso-words}]
Let $A_1, \dots, A_n$ be the deterministic weighted automata occurring in $E$.

First, let $x \in \conv(S_E)$. Then, $x = \sum_{i=1}^p \lambda_i v_i$ where $v_1,v_2,\ldots,v_p$ are the
vector values of simple cycles $\rho_1, \rho_2,\ldots, \rho_p$ in $A_E = A_1 \times \dots \times A_n$, 
and $\sum_{i=1}^p \lambda_i = 1$ with $\lambda_i \geq 0$ for all $1 \leq i \leq p$.

For each of the above cycles $\rho_i$, let $q_i$ be a state occurring in $\rho_i$, and let $\rho_{i \to j}$
be a simple path in $A_E$ connecting $q_i$ and $q_j$ (such paths exist for each $1 \leq i,j \leq p$
because $A_E$ has a unique strongly connected component). Let $\rho_{0 \to i}$ be a simple path in $A_E$
from the initial state $q_I$ to $q_i$. Note that the length of $\rho_i$ and $\rho_{i \to j}$ is at most 
$m = \abs{A_E}$ the number of states in $A_E$.
We consider the following sequence of ultimately periodic paths, 
parameterized by $N \in \nat$:
$$ \hat{\rho}_N = \rho_{0 \to 1} \cdot (\rho_1^{k_1^N} \cdot \rho_{1 \to 2} \cdot \ldots \cdot \rho_p^{k_p^N} \cdot \rho_{p \to 1})^{\omega},$$
where $k_i^N = \left\lfloor \frac{N \cdot \lambda_i}{\abs{\rho_i}}\right\rfloor$ for all  $1 \leq i \leq p$.
Note that $\hat{\rho}_N$ is the run of a lasso-word $w_N$ in $A_E$, and that 
$N \cdot \lambda_i -\abs{\rho_i}  \leq \abs{\rho_i} \cdot k_i^N \leq N \cdot \lambda_i$. 

Because $\hat{\rho}_N$ is ultimately periodic, the vector value of $\hat{\rho}_N$ gives the value of $w_N$ in each $A_i$. It can 
be computed as 
$$\LimAvg(\hat{\rho}_N) = \Avg(\rho_1^{k_1^N} \cdot \rho_{1 \to 2} \cdot \ldots \cdot \rho_p^{k_p^N} \cdot \rho_{p \to 1})$$
and it can be bounded along each coordinate $j = 1, \dots, n$ as follows  (we denote by $W$ the largest weight in $A_E$ in 
absolute value):

$$
\begin{array}{rcl}
\LimAvg_j(\hat{\rho}_N) & \leq & \frac{\displaystyle\sum_{i=1}^p   k_i^N \cdot \abs{\rho_i} \cdot \Avg_j(\rho_i) + 
                         \displaystyle\sum_{i=1}^p    \abs{A_E} \cdot W}
                     {\displaystyle\sum_{i=1}^p k_i^N \cdot \abs{\rho_i} + \displaystyle\sum_{i=1}^p \abs{A_E}} \\[+6pt]
             & \leq &  \frac{\displaystyle\sum_{i=1}^p   N \cdot \lambda_i \cdot \Avg_j(\rho_i) + 
                         p \cdot m \cdot W}
                     {\displaystyle\sum_{i=1}^p N \cdot \lambda_i - \abs{\rho_i} + \abs{A_E}} \\[+3pt]
            & \leq &  \frac{N \cdot x_j + p \cdot m \cdot W  {\large \strut} }
                     {N} =  x_j + \frac{p \cdot m \cdot W  {\large \strut} }
                     {N} \\[+3pt]
\end{array}
$$

Analogously, we have 
$$
\begin{array}{rcl}
\LimAvg_j(\hat{\rho}_N) & \geq & \frac{\displaystyle\sum_{i=1}^p   k_i^N \cdot \abs{\rho_i} \cdot \Avg_j(\rho_i) - 
                         \displaystyle\sum_{i=1}^p    \abs{A_E} \cdot W}
                     {\displaystyle\sum_{i=1}^p k_i^N \cdot \abs{\rho_i} + \displaystyle\sum_{i=1}^p \abs{A_E}} \\[+6pt]
             & \geq &  \frac{\displaystyle\sum_{i=1}^p   (N \cdot \lambda_i -  \abs{\rho_i})  \cdot \Avg_j(\rho_i) - 
                         p \cdot m \cdot W}
                     {\displaystyle\sum_{i=1}^p N \cdot \lambda_i + \abs{A_E} } \\[+3pt]
            & \geq &  \frac{N \cdot x_j - 2 p \cdot m \cdot W  {\large \strut} }
                     {N + p \cdot m} =  x_j - \frac{p \cdot m \cdot (2W - x_j)  {\large \strut} }
                     {N + p \cdot m} \\[+3pt]
\end{array}
$$

Therefore $\LimAvg_j(\hat{\rho}_N) \to x_j$ when $N \to \infty$. This shows that $x$ is in the closure
of the vector set of lasso-words.

Second, we show that the value of lasso words according to each automaton $A_i$ form a vector
which belong to $\conv(S_E)$ (which is equal to its closure). Let $w=w_1(w_2)^{\omega}$ be a 
lasso-word. 
It is easy to see that there exists $p_1,p_2$ such that $p = p_1 + p_2 \leq m = \abs{A_E}$ and the run of $A_E$ on $w_1 w_2^p$
has the shape of a lasso (i.e., the automaton $A_E$ is in the same state after reading $w_1 w_2^{p_1}$ and after
reading $w_1 w_2^p$), and thus the cyclic part of the lasso can be decomposed into simple cycles in $A_E$.
The vector value of $w$ in each $A_i$ is the mean of the vector values of the simple cycles in the decomposition,
and therefore it belongs to the convex hull $\conv(S_E)$.
\qed
\end{proof}

\begin{proof}[of Lemma~\ref{lem:value-set}]
First, show that $V_E \subseteq F_{\min}(\conv(S_E))$.
Let $x \in V_E$ be a tuple of values of some word $w$ according to each automaton $A_i$ occurring in $E$
(i.e., $x_i = L_{A_i}(w)$ for all $1 \leq i \leq n$). For $\epsilon > 0$ and $1 \leq k \leq n$, 
we construct a lasso-word $w^k_{\epsilon}$ such that $\abs{L_{A_k}(w_{\epsilon}) - x_k} \leq \epsilon$ and 
$L_{A_i}(w_{\epsilon}) \geq x_i - \epsilon$ for all $1 \leq i \leq n$ with $i \neq k$. If we denote
by $y^k_{\epsilon}$ the vector value of $w^k_{\epsilon}$, then the value $y = f_{\min}(\{y^k_{\epsilon} \mid 1 \leq k \leq n \})$
is such that $\abs{y_i - x_i} \leq \epsilon$ for all $1 \leq k \leq n$. By Lemma~\ref{lem:convex-hull-lasso-words},
the limit of the vector value $y^k_{\epsilon}$ when $\epsilon \to 0$ is in $\conv(S_E)$, and thus 
$x \in F_{\min}(\conv(S_E))$. 

We give the construction of $w^k_{\epsilon}$ for $k=1$. The construction is similar for $k \geq 2$. 
Consider the word $w$ and let $\rho$ be the suffix of the (unique) run of $A_E$ on $w$ which visits
only states in the strongly connected component of $A_E$. The value of $\rho$ and the value of $w$ 
coincide (according to each $A_i$) since the mean-payoff value is prefix-independent.
Since $L_{A_i}(w) = x_i$ for all $1 \leq i \leq n$,
there exists a position $p \in \nat$ such that the mean value of all prefixes of $\rho$ of length greater 
than $p$ is at least $x_i - \epsilon$ according to each $A_i$ (since each $A_i$ is a $\LimInfAvg$-automata).
Since $L_{A_1}(w) = x_1$, there exist infinitely many prefixes $\rho'$ of $\rho$ with mean value according to $A_1$ 
close to $x_1$, more precisely such that $\abs{\Avg_1(\rho') - x_1} \leq \epsilon$. Pick such a prefix $\rho'$ of length
at least $\max(p, \frac{1}{\epsilon})$. Since $\rho'$ is in the strongly connected component of $A_E$, 
we can extend $\rho'$ to loop back to its first state. This requires at most $m$ additional
steps and gives $\rho''$. Note also that $\rho''$ can be reached from the initial state of $A_E$ since it was the case of $\rho$,
and thus it defines a lasso-shaped run whose value can be bounded along the first coordinate as follows:
$$
\begin{array}{rcl}
\abs{\Avg_1(\rho'') - x_1} & \leq & \frac{ {\Large \strut} \bigabs{ \abs{\rho'} \cdot \Avg_1(\rho') - \abs{\rho''} \cdot x_1} + m \cdot W}{\abs{\rho''}}  \\[+3pt]
                           & \leq & \frac{ {\Large \strut} \abs{\rho'} \cdot \abs{\Avg_1(\rho') - x_1} + (\abs{\rho''} - \abs{\rho'}) \cdot x_1 + m \cdot W}{\abs{\rho''}}   \\[+3pt]
                           & \leq & \epsilon + \frac{ {\Large \strut} m \cdot x_1 + m \cdot W}{\abs{\rho''}} \leq \epsilon + \frac{ {\Large \strut} 2m \cdot W}{\abs{\rho''}}  \\[+5pt]
                           & \leq & \epsilon \cdot (1 + 2m \cdot W) \\[+1pt]
\end{array}
$$
Hence, the value along the first coordinate of the word $w^1_{\epsilon}$ corresponding to the run $\rho''$ tends to $x_1$ when when $\epsilon \to 0$. 
We show similarly that the value of $w^1_{\epsilon}$ along the other coordinates $i \geq 2$ is bounded from below by $x_i -  \epsilon \cdot (1 + 2m \cdot W)$.
The result follows. \medskip


Now, we show that $F_{\min}(\conv(S_E)) \subseteq V_E$.
In this proof, we use the notation $\odot$ for \emph{iterated concatenation} defined as follows. 
Given nonempty words $w_1,w_2 \in \Sigma^{+}$, the finite word $w_1 \odot w_2$ is $w_1 \cdot (w_2)^k$
where $k = \abs{w_1}^2$. We assume that $\odot$ (iterated concatenation) and $\cdot$ (usual concatenation) have the same 
precedence and that they are left-associative. For example, the expression $ab \odot a \cdot b$
is parsed as $(ab \odot a) \cdot b$ and denotes the word $abaaaab$, while the expression 
$ab \cdot a \odot b$ is parsed as $(ab \cdot a) \odot b$ and denotes the word $abab^9$.
We use this notation for the purpose of simplifying the proof presentation, and some care needs to be taken.
For example, explicit use of concatenation (i.e., $a \cdot b$ vs. $ab$) makes a difference
since $ab \odot ab = (ab)^5$ while $ab \odot a \cdot b = aba^4b$. Finally, we use notations
such as $(w_1 \cdot w_2 \,\odot)^{\omega}$ to denote the infinite word $w_1 \cdot w_2 \odot w_1 \cdot w_2 \odot \dots$.

Usually we use the notation $w_1 \odot w_2$ when the run of $A_E$ on $w_1 \cdot w_2$ can be decomposed as
$\rho_1 \cdot \rho_2$ where $\rho_i$ corresponds to $w_i$ ($i=1,2$) and $\rho_2$ is a cycle
in the automaton. Then, the mean value of the run on $w_1 \odot w_2$ is
\begin{xalignat*}{1}    
 & \frac{ \abs{\rho_1} \cdot \Avg(\rho_1) + \abs{\rho_1}^2 \cdot \abs{\rho_2} \cdot \Avg(\rho_2)}{ \abs{\rho_1}  + \abs{\rho_1}^2 \cdot \abs{\rho_2}}  \\[+6pt]
 = \ & \frac{ \Avg(\rho_1) + \abs{\rho_1} \cdot \abs{\rho_2} \cdot \Avg(\rho_2)}{ 1  + \abs{\rho_1} \cdot \abs{\rho_2}}  \\[+6pt]
 = \ & \Avg(\rho_2) + \frac{ \Avg(\rho_1) - \Avg(\rho_2)}{ 1  + \abs{\rho_1} \cdot \abs{\rho_2}} 
\end{xalignat*}
Therefore, since $\abs{ \Avg(\rho_1) - \Avg(\rho_2)} \leq 2W$ independently of $w_1$ and $w_2$, 
a key property of $\odot$ is that the mean value of $w_1 \odot w_2$ can be made 
arbitrarily close to $\Avg(\rho_2)$ by taking $w_1$ sufficiently long (since $\abs{w_1} = \abs{\rho_1}$).

We proceed with the proof of the lemma. 
Let $x \in F_{\min}(\conv(S_E))$ and let $y_1, \dots, y_n$ be $n$ points in $\conv(S_E)$ such that
the $i^{\text{th}}$ coordinate of $x$ and $y_i$ coincide for all $1 \leq i \leq n$, and 
the $j^{\text{th}}$ coordinate of $x$ is smaller than the $j^{\text{th}}$ coordinate of $y_i$ for all $j \neq i$.
Such $y_i$'s exist by definition of $F_{\min}$ though they may not be distinct. 

By Lemma~\ref{lem:convex-hull-lasso-words}, for all $\epsilon > 0$ there exist lasso-words 
$w_1, \dots, w_n$ such that $\norm{v_k - y_k} \leq \epsilon$ where $v_k = \tuple{L_{A_1}(w_k),\ldots,L_{A_n}(w_k)}$ 
for each $1 \leq k \leq n$. For each $1 \leq i \leq n$, let $\rho_i$ be the cyclic part of the (lasso-shaped) run 
of $A_E$ on $w_i$, and let $q_i$ be the first state in $\rho_i$. For each $1 \leq i,j \leq n$, define $\rho_{i \to j}$ the
shortest path in $A_E$ from $q_i$ to $q_j$, and let $\rho_{0 \to j}$ be a simple path in $A_E$ from the initial 
state $q_I$ to $q_j$ (such paths exist because $A_E$ is strongly connected). 
Note that $\Avg_j(\rho_i) = L_{A_j}(w_i)$.
We construct the following infinite run in $A_E$:
$$\hat{\rho} =  \rho_{0 \to 1} \odot (\rho_1 \cdot \rho_{1 \to 2} \odot \rho_2 \cdot \rho_{2 \to 3} \odot \dots \rho_n \cdot \rho_{n \to 1} \odot)^{\omega} $$
It is routine to show that $\hat{\rho}$ is a run of $A_E$, and 
we have $\LimAvg_j(\hat{\rho}) = v_{jj}$ because $(i)$ the cycles $\rho_1, \dots, \rho_n$ are asymptotically 
prevailing over the cycle $\rho_{1 \to 2} \rho_{2 \to 3} \dot \rho_{n \to 1}$, $(ii)$ 
by the key property of $\odot$, there exist infinitely many prefixes in  $\hat{\rho}$
such that the average of the weight along the $j^{\text{th}}$ coordinate converges to $v_{jj}$, 
and $(iii)$ all cycles $\rho_i$ have average value greater than $v_{jj}$ along the $j^{\text{th}}$
coordinate. Therefore, the liminf of the averages along the $j^{\text{th}}$ coordinate (i.e., $\LimAvg_j(\hat{\rho})$) 
is $v_{jj}$, and the vector of values of $\hat{\rho}$ is thus at distance $\epsilon$ of $x$, that is
$\norm{\LimAvg(\hat{\rho}) - x} \leq \epsilon$. The construction of $\hat{\rho}$ can be adapted to
obtain $\LimAvg(\hat{\rho}) = x$ by changing the $k^{\text{th}}$ occurrence of $\rho_i$ in $\hat{\rho}$ by a 
cycle corresponding to a lasso-word $w_i$ obtained as above for $\epsilon <  \frac{1}{n}$.
\qed
\end{proof}

\section{Proofs of Section~\ref{sec:alg-cons}}

\begin{proof}[of Lemma~\ref{lem:value-set-convex}]
Let $x = f_{\min}(u^1, u^2, \dots, u^n)$ and $y = f_{\min}(v^1, v^2, \dots, v^n)$
where $u^1, \dots, u^n, v^1, \dots, v^n \in X$. Let $z = \lambda x + (1-\lambda) y$ where $0 \leq \lambda \leq 1$ 
and we prove that $z \in F_{\min}(X)$.
Without loss of generality, assume that $x_i = u^i_i$ and $y_i = v^i_i$ for all $1 \leq i \leq n$. 
Then $z_i =  \lambda u^i_i + (1-\lambda) v^i_i$ for all $1 \leq i \leq n$.

To show that $z \in F_{\min}(X)$, we give for each $1 \leq j \leq n$ a point $p \in X$ such that 
$p_j = z_j$ and $p_k \geq z_k$ for all $k \neq j$. Take $p = \lambda u^j + (1-\lambda) v^j$. 
Clearly $p \in X$ since $u^j, v^j \in X$ and $X$ is convex, and $(i)$ $w_j = \lambda u^j_j + (1-\lambda) v^j_j = z_j$,
and $(ii)$ for all $k \neq j$, we have $w_k =  \lambda u^j_k + (1-\lambda) v^j_k \geq \lambda u^k_k + (1-\lambda) v^k_k = z_k$
(since $u^k$ has the minimal value on $k^{\text{th}}$ coordinate among $u^1, \dots, u^n$, similarly for $v^k$). 
\qed
\end{proof}

\begin{proof}[of Proposition~\ref{prop:conv-min-2D}]
By Lemma~\ref{lem:value-set-convex}, we already know that $\conv(F_{\min}(S)) \subseteq F_{\min}(\conv(S))$
(the set $F_{\min}(\conv(S))$ is convex, and since 
$F_{\min}$ is a monotone operator and $S \subseteq \conv(S)$, we have  $F_{\min}(S) \subseteq F_{\min}(\conv(S))$
and thus $\conv(F_{\min}(S)) \subseteq F_{\min}(\conv(S))$).

We prove that $F_{\min}(\conv(S)) \subseteq \conv(F_{\min}(S))$ if $S \subseteq \real^2$. 
Let $x \in F_{\min}(\conv(S))$ and show that $x \in \conv(F_{\min}(S))$. 
Since $x \in F_{\min}(\conv(S))$, there exist $p,q \in \conv(S)$ such that 
$x = f_{\min}(p,q)$, and assume that $p_1 < q_1$ and $p_2 > q_2$ (other cases are 
symmetrical, or imply that $x=p$ or $x=q$ for which the result is trivial as then $x \in \conv(S)$).
We show that $x = (p_1,q_2)$ is in the convex hull of $\{p,q,r\}$ where $r = f_{\min}(u,v)$ 
and $u \in S$ is the point in $S$ with smallest first coordinate, and $v \in S$ is the point 
in $S$ with smallest second coordinate, so that $r_1 = u_1 \leq p_1$ and $r_2 = v_2 \leq q_2$.
Simple computations show that the equation $x = \lambda p + \mu q + (1-\lambda-\mu) r$ has
a solution with $0 \leq \lambda, \mu  \leq 1$ and the result follows.
\qed
\end{proof}

\begin{proof}[of Lemma~\ref{lemm1}]
By definition, we have $F_n(S) \subseteq F(S)$.
For a point $x=f(P)$ for a finite subset $P \subseteq S$, choose one point each that 
contributes to a coordinate and obtain a finite set $P' \subseteq P$ of at most $n$
points such that $x=f(P)$. This shows that $F(S) \subseteq F_n(S)$.

For the second part, let $P=\set{p_1, p_2,\ldots,p_k}$ with $k \leq n$, and let $x=f(P)$.
Let $x_1=f(p_1,p_2)$, and for $i > 1$ we define $x_i=f(x_{i-1},p_{i+1})$.
We have $x=x_{n-1}$ (e.g., $f(p_1,p_2,p_3)= f (f(p_1,p_2), p_3)$).
Thus we have obtained $x$ by applying $f$ on two points for $n-1$ times, and
it follows that $F_n(S) \subseteq F_2^{n-1}(S)$.
\qed
\end{proof}

\begin{proof}[of Theorem~\ref{theo:explicit-construction}]
We show that the construction $\gamma$ satisfies condition {\bf C1} and {\bf C2}.
Let $Y' = \gamma(Y)$. 
Clearly the set $Y'$ is a finite subset of $\convk(Y)$ and thus Condition {\bf C1} 
holds and we now show that Condition {\bf C2} is satisfied.

%
%
Since $F_2(\conv(Y))$ is convex (by Lemma~\ref{lem:value-set-convex}), it suffices
to show that all corners of $F_2(\conv(Y))$ belong to $\convk(F(Y'))$.
Consider a point $x=f(p,q)$ where $p,q \in \convk(Y)$. 
We will show that either $p,q \in Y'$ or $x$ 
cannot be a corner of $\convk(F_2(Y))$. It will follow that 
$F_2(\conv(Y)) \subseteq \convk(F(Y'))$. 
%
%
Our proof will be an induction on the number of coordinates such that there is a 
\emph{tie} (tie is the case where the value of a coordinate of $p$ and $q$ coincide).
If there are $n$ ties, then the points $p$ and $q$ are equal and we have $x=p=q$, 
and this case is trivial since $Y \subseteq Y'$. So the base case is done. 
By inductive hypothesis, we assume that $k+1$-ties yield the result and we 
consider the case for $k$-ties.
Without loss of generality we consider the following case:
\[
p_1=q_1; p_2=q_2; \cdots; p_k=q_k;
\]
\[
p_{k+1} < q_{k+1}; p_{k+2} < q_{k+2}; \cdots; p_{\ell}< q_{\ell};
\]
\[
p_{\ell+1} > q_{\ell+1}; p_{\ell+2} > q_{\ell+2}; \cdots; p_{n}> q_{n};
\]
i.e, the first $k$ coordinates are ties, then $p$ is the sole contributor 
to the coordinates $k+1$ to $\ell$, and for the rest of the coordinates $q$ is 
the sole contributor.
Below we will use the expression \emph{infinitesimal change} to mean change smaller than 
$\eta =\min_{k <i \leq n} \abs{p_i -q_i}$ (note $\eta>0$).
Consider the plane $\Pi$ with first $k$ coordinates constant (given by 
$x_1=p_1=q_1; x_2=p_2=q_2; \cdots; x_k=p_k=q_k$).
We intersect the plane~$\Pi$ with $\convk(Y)$ and we obtain a polytope. 
First we consider the case when $p$ and $q$ are not a corner of the polytope 
and then we consider when $p$ and $q$ are corners of the polytope.

\begin{enumerate}
\item \emph{Case 1: $p$ is not a corner of the polytope $\Pi \cap \convk(Y)$.}
We draw a line in $\Pi$ with $p$ as midpoint such that the line is contained 
in $\Pi \cap \convk(Y)$. This ensures that the coordinates $1$ to $k$ remain 
fixed along the line.
\begin{enumerate}
\item If any one of coordinates from $k+1$ to $\ell$ changes along the line, then by 
infinitesimal change of $p$ along the line, we ensure that $x$ moves along a line.

\item Otherwise coordinates $k+1$ to $\ell$ remain constant; and we move~$p$
along the line in a direction such that at least one of the remaining coordinates (say $j$) 
decreases, and decreasing $j$ we have one of the following three cases:
\begin{enumerate}
\item we go down to $q_j$ and then we have one more tie and we are fine by 
inductive hypothesis;

\item we hit a face of the polytope $\Pi \cap \convk(Y)$ and then we change direction of the line (while staying in the hit face) and continue;

\item we hit a corner of the polytope $\Pi \cap \convk(Y)$ and then $p$ becomes a corner
which will be handled in Case 3.
\end{enumerate}

\end{enumerate} 

\item \emph{Case 2: $q$ is not a corner of the polytope $\Pi \cap \convk(Y)$.}
By symmetric analysis to Case~1 either we are done or $q$ becomes a corner 
of the polytope $\Pi \cap \convk(Y)$. 

\item \emph{Case 3: $p$ and $q$ are corners of the polytope $\Pi \cap \convk(Y)$.}
If $\Pi$ is supported by $Y$, then both $p,q \in Y'$ and we are done.
Otherwise $\Pi$ is not supported by $Y$, and now we move along lines with 
$p$ and $q$ as midpoints and slide the plane $\Pi$. 
In other words we move $p$ and $q$ alone lines and move such that the ties 
remain the same.
We also ensure infinitesimal changes along the line so that the contributor 
of each coordinate is the same as original.
Let 
\[
p(\lambda) =p + \lambda \cdot \vec{v};
\quad
q(\mu) =q  + \mu \cdot \vec{w};
\]
be the lines where $\vec{v}$ and $\vec{w}$ are directions.
By ties for $1 \leq i \leq k$ we have $\lambda \cdot v_i =\mu \cdot w_i$.
Then for infinitesimal change the point $x$ moves as follows: 
\[
\begin{array}{l}
x(\lambda,\mu)  =    f(p(\lambda),q(\mu))) \\[1ex]
=  
(p_1 + \lambda \cdot v_1, 
p_2 + \lambda \cdot v_2, \cdots
p_\ell + \lambda \cdot v_\ell, 
q_{\ell +1} + \mu \cdot w_{\ell+1}, \cdots, 
q_{n} + \mu \cdot w_{n}) 
\\[1ex]
= (p_1 + \lambda \cdot v_1, 
p_2 + \lambda \cdot v_2, \cdots
p_\ell + \lambda \cdot v_\ell, 
q_{\ell +1} + \lambda \cdot\frac{v_1}{w_1} \cdot w_{\ell+1}, \cdots, 
q_{n} + \lambda \cdot\frac{v_1}{w_1} \cdot w_{n}) 
\end{array}
\]
It follows that $x$ moves along the line 
$x + \lambda \cdot \vec{z}$ where for 
$1 \leq i \leq \ell$ we have $z_i=v_i$ and for $\ell < i \leq n$ we have 
$z_i= \frac{v_1}{w_1}\cdot w_i$; note that $w_1 >0$ since the plane slides.
Since $x$ moves along a line it cannot be an extreme point.
\end{enumerate}
This completes the proof. Also note that in the special case when there is 
no tie at all then we do not need to consider Case 3 as then $\Pi = \real^n$ 
and thus $p$ and $q$ are corners of $\convk(Y)$ and hence in $Y'$.

\smallskip\noindent{\bf Analysis.} Given a set of $m$ points, 
the construction $\constr$ yield at most $m^2 \cdot 2^n$ points. 
The argument is as follows: consider a point $p$, and then we consider 
all $k$-dimensional coordinates planes through $p$. 
There are ${n \choose k}$ possible $k$-dimensional coordinate plane through $p$,
and summing over all $k$ we get that there are at most $2^n$ coordinate planes 
that we consider through $p$.
The interesection of a coordinate plane through $p$ with the convex hull of 
$m$ points gives at most $m$ new corner points, and this claim is as proved 
follows: the new corner points can be constructed as the shadow of the convex
hull on the plane, and since the convex hull has $m$ corner points the claim
follows.
Thus it follows that the construction yield at most $m^2 \cdot 2^n$ new points,
and thus we have at most $m+ m^2 \cdot 2^n \leq 2 \cdot m^2 \cdot 2^n$ points. 
If the set $S$ has $m$ points, applying the construction iteratively for $n$ 
times we obtain the desired set $S'$ that has at most 
$m^{2^n} \cdot 2^{n^2+n}$ points.
Since convex hull of a set of $\ell$ points in $n$ dimension can be 
constructed in $\ell^{O(n)}$ time, it follows that the set $S'$ can 
be constructed in $m^{O(n\cdot 2^n)} \cdot 2^{O(n^3)}$ time. 
\qed
\end{proof}

\begin{proof}[Theorem~\ref{thrm_undecidable} (Sketch)] 
We will show the undecidability for the quantitative universality problem 
for nondeterministic mean-payoff automata. 
It will follow that the quantitative language inclusion and quantitative 
language equivalence problem are undecidable for both nondeterministic and 
alternating automata. 
The quantitative universality for nondeterministic automata can be reduced to 
the quantitative emptiness as well as the quantitative universality problem 
for alternating mean-payoff automata. Hence to complete the proof we derive 
the undecidability of quantitative universality for nondeterministic 
mean-payoff automata from the recent results of~\cite{DDGRT10}.

The results of~\cite{DDGRT10} show that in two-player \emph{blind} 
imperfect-information mean-payoff games whether there is a player~1 
blind-strategy $\sigma$ such that against all player~2 strategies $\tau$ the 
mean-payoff value $P(\sigma,\tau)$ of the play given $\sigma$ and $\tau$ is 
greater than $\nu$ is undecidable.
The result is a reduction from the halting problem of two-counter machines, 
and we observe that the reduction has the following property: 
for threshold value $\nu=0$, if the two-counter machine halts then 
player~1 has a blind-strategy to ensure payoff greater than $\nu$, and otherwise
against every blind-strategy for player~1, player~2 can ensure that the payoff for 
player~1 is at most $\nu=0$. 
Thus from the above observation about the reduction of~\cite{DDGRT10} it follows 
that in two-player blind imperfect-information mean-payoff games, given a 
threshold $\nu$, the decision problem whether 
\[
\exists \sigma.\ \inf_{\tau} P(\sigma,\tau) > \nu
\]
where $\sigma$ ranges over player~1 blind-strategies, and $\tau$ over player~2 
strategies, is undecidable and dually the following decision problem whether
\[
\forall \sigma.\ \sup_{\tau} P(\sigma,\tau) \geq \nu
\]
is also undecidable.
The universality problem for nondeterministic mean-payoff automata is equivalent
to two-player blind imperfect information mean-payoff games where the choice of 
words represents the blind-strategies for player~1 and resolving nondeterminism 
corresponds to strategies of player~2. 
It follows that for nondeterministic mean-payoff automata $A$, 
given a threshold $\nu$, the decision problem whether 
\[
\text{for all words } w. \ L_A(w) \geq \nu
\]
is undecidable.
\qed
\end{proof}

\section{Proofs of Section~\ref{sec:expressive-power}}

\begin{proof}[of Theorem~\ref{thrm_expressive_power}]
We prove the two assertions.
\begin{enumerate}
\item The results of~\cite{CDH-FCT09} shows that there exists deterministic mean-payoff 
automata $A_1$ and $A_2$ such that $\Sum(A_1,A_2)$ cannot be expressed by alternating mean-payoff 
automata.
Hence the result follows.

\item We now show that there exist quantitative languages expressible by 
nondeterministic mean-payoff automata that cannot be expressed by mean-payoff 
automaton expressions.
Consider the language $L_F$ of finitely many $a$'s, i.e., for an infinite word $w$ 
we have $L_F(w)=1$ if $w$ contains finitely many $a$'s, and $L_F(w)=0$ otherwise. 
It is easy to see that the nondeterministic mean-payoff automaton (shown in 
\figurename~\ref{figure:aut3}) defines $L_F$. 
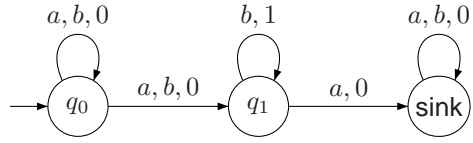
\begin{figure}[!t]
 \begin{center}
 \unitlength=.8mm
\def\fsize{\normalsize}

\begin{picture}(78,30)(0,0)

{\fsize

\node[Nmarks=i](x0)(12,10){$q_0$}
\node[Nmarks=n](x1)(42,10){$q_1$}
\node[Nmarks=n, Nadjust=w](x2)(72,10){{\sf sink}}

\drawloop[ELside=l, ELdist=1, loopCW=y, loopdiam=7, loopangle=90](x0){$a,b,0$}
\drawedge[ELpos=50, ELside=l, ELdist=1, curvedepth=0](x0,x1){$a,b,0$}
\drawloop[ELside=l, ELdist=1, loopCW=y, loopdiam=7, loopangle=90](x1){$b,1$}
\drawedge[ELpos=50, ELside=l, ELdist=1, curvedepth=0](x1,x2){$a,0$}
\drawloop[ELside=l, ELdist=1, loopCW=y, loopdiam=7, loopangle=90](x2){$a,b,0$}



}
\end{picture}
 \end{center}
 \caption{A nondeterministic limit-average automaton. \label{figure:aut3}}
\end{figure}

We now show that $L_F$ is not expressible by a mean-payoff automaton expression.
Towards contradiction, assume that the expression $E$ defines the language $L_F$,
and let $A_E$ be the synchronized product of the deterministic automata occurring
in $E$ (assume $A_E$ has $n$ states). 
Consider a reachable bottom strongly connected component $V$ of the underlying graph
of $A_E$, and let $C$ be a $b$-cycle in $V$. We construct an infinite word $w$ 
with infinitely many $a$'s as follows: $(i)$ start with a prefix $w_1$ of length at most $n$ to reach $C$, $(ii)$
loop $k$ times through the b-cycle $C$ (initially $k=1$), $(iii)$ read an `$a$' and then a finite word
of length at most $n$ to reach $C$ again (this is possible since $C$ is in a bottom s.c.c.), 
and proceed to step $(ii)$ with increased value of $k$.

The cycle $C$ corresponds to a cycle in each automaton of $E$, and
since the value of $k$ is increasing unboundedly, the value of $w$ in
each automaton of $E$ is given by the average of the weights along
their $b$-cycle after reading $w_1$. Therefore, the value of $w$ and the
value of $w_1 b^\omega$ coincide in each deterministic automaton of $E$.
As a consequence, their value coincide in $E$ itself. This is a contradiction
since $L_F(w)=0$ while  $L_F(w_1 b^\omega) = 1$.
\end{enumerate}
\qed
\end{proof}

\end{document}